\documentclass[11pt, draftclsnofoot, onecolumn]{IEEEtran}
\usepackage{multirow, verbatim, amsfonts, graphicx, amsmath, amsbsy, amssymb,  epsfig, url}
\newcommand{\ab}{{\mathbf a}}
\newcommand{\bb}{{\mathbf b}}

\newcommand{\nb}{{\mathbf n}}

\newcommand{\xb}{{\mathbf x}}

\newtheorem{definition}{Definition}

\newtheorem{theorem}{Theorem}

\title{Exact Dynamic Support Tracking with Multiple Measurement Vectors using Compressive MUSIC}

\author{Jong Min Kim, Ok Kyun Lee and Jong Chul Ye}

\begin{document}
%
\maketitle
\begin{abstract}
\baselineskip 0.18in
Dynamic tracking of sparse targets has been one of the important topics in array signal processing. Recently, compressed sensing (CS) approaches have been extensively investigated as a new tool for this problem using partial support information obtained by exploiting temporal redundancy. However, most of these approaches are formulated under single measurement vector   compressed sensing (SMV-CS) framework, where the performance guarantees are only in a probabilistic manner. The main contribution of this paper is to  allow \textit{deterministic} tracking of time varying supports with multiple measurement vectors (MMV) by exploiting multi-sensor diversity. In particular, we show that  a novel compressive MUSIC (CS-MUSIC) algorithm with optimized partial support selection not only  allows removal of inaccurate portion of previous support estimation but also  enables addition of  newly emerged part of unknown support. Numerical results confirm the theory.
\end{abstract}
\baselineskip 0.18in

\vspace*{1cm}

\noindent{ 
\\

Correspondence to:\\
Jong Chul Ye,  Ph.D.
Associate Professor \\\
Dept. of Bio and Brain Engineering,  KAIST \\
373-1 Guseong-dong  Yuseong-gu, Daejon 305-701, Korea \\
Email: jong.ye@kaist.ac.kr \\
Tel: 82-42-350-4320 \\
Fax: 82-42-350-4310 \\
 }

\IEEEpeerreviewmaketitle

\newpage
\baselineskip 0.29in

\section{Introduction}
\label{sec:intro}

\IEEEPARstart{D}{ynamic} target tracking problem that addresses the estimation of   time varying support of moving target  has  been one of the important classical topics in array signal processing including radar,  communication, and medical imaging applications. For example, in electroencephalography (EEG) or magnetoencephalography (MEG) source localization problems, it has been shown that the position of the dipole moments during epileptic activities varies according to time  and we are interested in their spatio-temporal dynamics \cite{Zhang2011TSBL}. Dynamic MRI problem that tracks the motion of hearts also belongs to this class of problem.

Recently, there  have been renewed interests for this problem with the help of a modern mathematical tool called compressed sensing \cite{Do06,CaRoTa06}. These approaches try to exploit knowledges of partial support information obtained  at the previous time point.
More specifically, consider the following time varying support estimation problem:
\begin{eqnarray}\label{eq:timevaryingSMV}
\min\limits_{\xb(t)}~\|\xb(t)\|_0,~~{\rm subject~to}~~\bb(t)=A\xb(t) ,~~t=0,1,\cdots,
\end{eqnarray}
where $\bb(t) \in \mathbb{R}^{m}$, and $\xb(t) \in \mathbb{R}^{n}$ are noiseless measurement vector, and sparse signal at time $t$. Assuming that the support is assumed to change slowly, theoretical results \cite{Vaswani2010MCS} have demonstrated that we can reduce the required sampling in compressed sensing reconstruction  if we have partially known support from the prior estimation results. For example, Vaswani and Lu proposed modified-CS algorithm \cite{Vaswani2010MCS} which addresses the exact reconstruction of noiseless case with partially known support:
\begin{eqnarray}\label{eq:timevaryingSMVmodl0}
\min\limits_{\xb(t)}~\|(\xb(t))_{{I{(t-1)}}^c}\|_0,~~{\rm subject~to}~~\bb(t)=A\xb(t),~~t=1,2,\cdots,
\end{eqnarray}
where $I(t-1)$ is the previously estimated support, and $(\xb(t))_{{I{(t-1)}}^c}$ denotes a subvector after removing the elements that correspond to the index set $I(t-1)$. Suppose, furthermore,  $k=|{\rm supp}\xb(t)|_0$, $u=|I(t)\setminus I(t-1)|$, and $e=|I(t-1) \setminus I(t)|$. Then, if the restricted isometry constant (RIP)  for the sensing matrix $A$ satisfies
\begin{eqnarray}\label{eq:timevaryingSMVmodl0RIP}
\delta_{k+e+u}<1,
\end{eqnarray}
then the solution $\xb(t)$ of Eq.~\eqref{eq:timevaryingSMVmodl0} is the unique solution \cite{Vaswani2010MCS}. This is much weaker than $0\leq \delta_{2k}<1$ for the original SMV-CS problem \cite{CaTa05}, in case of  slowly time varying support with $u \ll k$ and $e \ll k$.
They further showed an $l_1$ convex relaxation of  Eq.~\eqref{eq:timevaryingSMVmodl0}
 can provide the same $l_0$ solution of Eq.~\eqref{eq:timevaryingSMVmodl0}, if the following RIP condition is satisfied:
\begin{eqnarray}\label{eq:timevaryingSMVmodl1RIP}
2\delta_{2u} + \delta_{3u} + \delta_{k+e-u} + \delta_{k+e}^2 + 2\delta_{k+e+u}^2 < 1 ,
\end{eqnarray}
which is again relaxed sampling requirement than that of original CS problem $\delta_{2k}<\sqrt{2}-1$ \cite{CaTa05}.  Therefore, exploiting the temporal redundancy has significant impact for reducing sampling requirement for dynamic support tracking. 

Rather than solving the tracking problem Eq.~\eqref{eq:timevaryingSMV} at each time, batch type approaches such as  T-SBL (temporal sparse Bayesian learning) \cite{Zhang2011TSBL} collect the multiple snapshot data (for example, $\{\bb(t)\}_{t=1}^N$) and process them together to estimate the dynamic varying support.  Note that if the support changes slowly over time, then the resulting collection of problem becomes an multiple measurement vector problem. Accordingly,  T-SBL  converts the resulting MMV problem into a block-sparse SMV problem, after which each block statistics are modeled using a specific Gaussian form temporal correlation structure. The update rule using the expectation-maximization (EM) method and its accelerated version  can be then used to solve the resulting Bayesian problem \cite{Zhang2011TSBL}.

However, these  approaches for dynamic support tracking is with SMV-CS framework  and their performance guarantees is  in a probabilistic sense.
In practice, there are many situations where we can obtain multiple measurement vector information for time varying objects. 
For example, in single-input multiple-output (SIMO) multiple access channel (MAC), multiple antenna can observe linear combination of individual codewords multiplied by the unknown channel gain from the individual user \cite{jin2011support}.  In parallel MR cardiac imaging,  multiple coils  simultaneous obtain k-space measurements of temporally varying hearts  with distinct coil sensitivities.  In EEG/MEG source localization problem, the dipole moments can be assumed relatively stationary during a short time window from which multiple snapshot of the sensor measurement can be obtained. All these examples acquire multiple measurement of the unknown signal vectors that share the same support with different weighting through identical sensing matrices.

A fundamental question under this setup is what kind of diversity gain we can obtain over SMV-CS support tracking. To our knowledge, we are not aware of any prior investigation in this regard.
One of the main contributions of this paper is to show that a multiple measurement vector (MMV) framework not only extend the SMV counterpart, but also provides a unique advantage of ``deterministic'' support tracking for slow varying support estimation. Recall that MMV can measure multiple information of a set of vector that share the same sparsity pattern through the identical sensing matrix. This paper shows that this joint sparsity pays off significantly in dynamic support tracking by relaxing probabilistic guarantee to a deterministic guarantee. The feasibility of the exact support tracking has significant impacts in practice. 

The breakthrough  is based on our novel compressive multiple signal classification  (CS-MUSIC)  algorithm in MMV compressed sensing problem \cite{Kim2010CMUSIC}, in which a part of supports are found probabilistically using the conventional CS, after which the remaining supports are determined deterministically using the generalized MUSIC criterion. In addition, CS-MUSIC allows us to find all $k$ support   as long as at least $k-r+1$ support out of any $k$-support estimate are correct \cite{Kim2011CSMUSIC}, where $r$ denote the rank of the measurement matrix. This result provides an important clue for deterministic and exact dynamic support tracking under MMV setup, in which the probabilistic compressed sensing support estimation step is replaced by the support estimate from the previous snapshots, after which the CS-MUSIC algorithm eliminates the incorrect portion of previous time point support  estimation and then add newly updated support deterministically. This update scheme guarantees the exact support tracking in noiseless case under an appropriate sampling condition. 
Other contributions of our method include that the support error does not propagate along time due to the self-correction step. Furthermore, using large system model, we can derive conditions with which the proposed algorithm correct track the time varying support even in noisy cases. We believe that with these noticeable advantages of our algorithm we may find many important applications in radar, communication as well as biomedical application.

This paper consist of following. Section \ref{sec:review} reviews the compressive MUSIC and support correction criterion for MMV setup. In Section \ref{sec:proposedMethod}, we derive our main theoretical results on sampling condition for deterministic support tracking. Numerical results are given in Section \ref{sec:results}, which is followed by conclusion in Section \ref{sec:conclusion}.


\subsection{Notations and Mathematical Preliminaries}

Throughout the paper, $\xb^i$ and $\xb_j$ correspond to the $i$-th row and the $j$-th column of matrix $X$, respectively. When $S$ is an index set, $X^S$, $A_S$ corresponds to a submatrix collecting corresponding rows of $X$ and columns of $A$, respectively. The following definitions are also used throughout the paper.
\begin{definition}\cite{Donoho20204npss}
The rows (or columns) in $\mathbb{R}^n$ are in general position if any $n$ collection of rows (or columns) are linearly independent.
\end{definition}

\begin{definition}\cite{DoEl03}
${\rm Spark}(A)$ denotes the smallest number of linearly dependent columns of a matrix $A$.
\end{definition}

\begin{definition}[Restricted Isometry Property (RIP)]
A sensing matrix $A\in\mathbb{R}^{m\times n}$ is said to have a $k$-restricted isometry property (RIP) if there exist left and right RIP constants $0<\delta^L_k, \delta^R_k<1$  such that
$$(1-\delta^L_k)\|\mathbf{x}\|^2\leq \|A\mathbf{x}\|^2\leq (1+\delta^R_k)\|\mathbf{x}\|^2$$
for all $\mathbf{x}\in\mathbb{R}^n$ such that $\|\mathbf{x}\|_0\leq k$. A single RIP constant $\delta_k = \max\{\delta^L_k, \delta^R_k\}$ is often referred to as the RIP constant.
\end{definition}

\section{MMV Compressive Sensing using Compressive MUSIC: A Review}
\label{sec:review}

Let $m$, $n$ and $r$  be a positive integers ($m<n$) that represents
the number of sensor elements, the ambient space dimension, and the
number of snapshots, respectively. Suppose that we are given a
multiple-measurement vector $B\in\mathbb{R}^{m\times r}$,
$X=[\mathbf{x}_1,\cdots,\mathbf{x}_r]\in\mathbb{R}^{n\times r}$, and
a sensing matrix $A\in \mathbb{R}^{m\times n}$.
A canonical form MMV problem \cite{Kim2010CMUSIC}
is given by the following optimization problem:
\begin{eqnarray}\label{eq:canMMV}
{\rm minimize}~~~\|X\|_0\\
{\rm subject~to}~~~B=AX \notag,
\end{eqnarray}
where $\|X\|_0=|{\rm supp}X|=k$, ${\rm supp}X=\{1\leq i\leq n :
\mathbf{x}^i\neq 0\}$, and the measurement matrix $B$ is full rank, i.e. ${\rm rank}(B)=r\leq \|X\|_0$.

Recall that every MMV problem can be converted to a canonical form MMV using a singular value decomposition and dimension reduction as described in \cite{Kim2010CMUSIC}.
Now, We can easily expect that the diversity due to the joint sparsity can improve the recovery performance over SMV compressed sensing. Indeed, Chen and Huo \cite{chen2006trs}, Feng and Bresler \cite{Feng97}  and recently Davies and Elder \cite{Davies2010Rank} showed that
$X\in\mathbb{R}^{n\times r}$  is the unique solution of $AX=B$ if and only if
\begin{equation}\label{eq:l0-bound-mmv}
\|X\|_0<  \frac{{\rm spark}(A)+{\rm rank}(B)-1}{2}
\leq {\rm spark}(A)-1  \ .
\end{equation}
Note that we can expect ${\rm rank}(B)/2$ gains over SMV thanks to the MMV diversity. Furthermore, Feng and Bresler \cite{Feng97}   showed that the noiseless $l_0$ bound in Eq.~\eqref{eq:l0-bound-mmv} is achievable using MUSIC algorithm as long as $r={\rm rank}(B)=k$. More specifically, suppose that the columns of a sensing matrix $A\in\mathbb{R}^{m\times n}$ are in general position. Then, according to \cite{Feng97,schmidt1986multiple},  for any $j\in \{1,\cdots,n\}$, $j\in {\rm supp}X$ if and only if
\begin{equation}\label{eq:musicCond}
Q^{*}\mathbf{a}_j=0,
\end{equation}
where $Q\in\mathbb{R}^{m\times (m-r)}$ consists of orthonormal columns such that $Q^{*}B=0$ so that $R(Q)^{\perp}=R(B)$, which is often called ``noise subspace''. Using the compressive sensing terminology, Eq.~\eqref{eq:musicCond} implies that the recoverable sparsity level by MUSIC (with a probability $1$ for the noiseless measurement case) is given by
\begin{equation}\label{eq:max_music}
    \|X\|_0 < m = {\rm spark}(A)-1,
\end{equation}
where the last equality comes from the definition of the ${\rm spark}$. Therefore, the $l_0$ bound \eqref{eq:l0-bound-mmv} can be achieved by MUSIC bound in \eqref{eq:max_music} when $r=k$  \cite{Feng97}.

However, for any $r<k$, the MUSIC condition \eqref{eq:musicCond} does not hold. This is a major drawback of MUSIC
compared to CS algorithms that allow perfect
reconstruction with a extremely large probability by increasing the
sensor elements $m$.  One the other hand, even thought the conventional CS algorithms for MMV such as simultaneous OMP (S-OMP), $p$-thresholding \cite{gribonval2008aac, Eldar2010aca} have good recovery performance when $r\ll k$, but they exhibit performance saturation as $r$ increases and never achieve the $l_0$ bound with finite snapshot even in noiseless case. Recently, we showed that this drawback of the existing approaches can be overcome by the following generalized MUSIC criterion \cite{Kim2010CMUSIC}.

\begin{theorem}\label{com-music} \cite{Kim2010CMUSIC}
Assume that $A\in\mathbb{R}^{m\times n}$, $X\in\mathbb{R}^{n\times
r}$, and $B\in\mathbb{R}^{m\times r}$ satisfy $AX=B$. Furthermore, we assume that
$\|X\|_0=k$ and $A$ satisfies the RIP condition with the left RIP constant $0<\delta^L_{2k-r+1}<1$.
If we are given $I_{k-r}\subset {\rm supp}X$
with $|I_{k-r}|=k-r$ and $A_{I_{k-r}}\in \mathbb{R}^{m\times(k-r)}$,
which consists of columns whose indices are in $I_{k-r}$, then for
any $j\in \{1,\cdots,n\}\setminus I_{k-r}$,
\begin{equation}\label{eq-comusicmod}
\ab_j^{*}\left[P_{R(Q)}-P_{R(P_{R(Q)}A_{I_{k-r}})}\right]\ab_j=0
\end{equation}
if and only if $j\in {\rm supp}X$.
\end{theorem}
In \cite{Kim2010CMUSIC}, we demonstrate that the condition $0<\delta^L_{2k-r+1}<1$ for generalized MUSIC is equivalent to $l_0$ bound \eqref{eq:l0-bound-mmv}, which implies that a computational expensive combinatorial optimization problem is now reduced to $|I_{k-r}|$ support estimation from the original $|I_{k}|$ support estimation\footnote{When $r=k$, the
condition (\ref{eq-comusicmod}) is the same as the MUSIC criterion (\ref{eq:musicCond}) and no combinatorial algorithm is necessary.}. Furthermore, by Theorem \ref{com-music}, we can develop a computationally tractable relaxation algorithm called Compressive MUSIC (CS-MUSIC) that relaxes the combinatorial optimization step of finding $I_{k-r}$ support using the conventional MMV-CS algorithms \cite{Kim2010CMUSIC}. The algorithm can be stated as following:
\begin{itemize}
   \item {\bf (Step 1: compressed sensing step)} Find $k-r$ indices of ${\rm supp}X$ by any MMV compressive sensing algorithms such as 2-thresholding or SOMP. 
 Let $I_{k-r}$ be set of   selected indices and $S=I_{k-r}$.
   \item   {\bf (Step 2: generalized MUSIC step)}   For $j\in \{1,\cdots,n\}\setminus I_{k-r}$,  calculate the quantities $\eta(j)=\ab_j^{*}[P_{R(Q)}-P_{R(P_{R(Q)}A_{I_{k-r}})}]\ab_j$ for all $j\notin
   I_{k-r}$.
  Make an ascending ordering of $\eta(j)$, $j\notin
   I_{k-r}$ and choose indices that correspond to the first $r$
   elements and put these indices into $S$.
\end{itemize}

In compressive MUSIC, we determine $k-r$ indices of ${\rm supp}X$
with CS-based algorithms such as 2-thresholding or S-OMP rather than $l_0$ optimization, where the
exact identification of $k-r$ indices is a probabilistic matter. After that process,
we recover remaining $r$ indices of ${\rm supp}X$ with a generalized
MUSIC criterion, which is given in Theorem \ref{com-music}, and this reconstruction process is
deterministic. This hybridization makes the compressive MUSIC
applicable for all ranges of $r$, outperforming all the existing methods.
Similar observation have been made independently by Lee and Bresler  \cite{Lee2010SAMUSIC} in their subspace augmented MUSIC (SA-MUSIC) algorithm.

To analyze the performance of the compressive MUSIC,
we should  find the number of measurements with which we can identify the support of $X$. Due to the reduction of uncertainty from $|I_k|$ to $|I_{k-r}|$, we can expect more relaxed sampling condition. In \cite{Kim2010CMUSIC}, we derived the sampling requirements when subspace S-OMP or 2-thresholding is used as a compressed sensing step for compressive MUSIC. The results can be summarized as following.
The number of measurements for subspace S-OMP for partial support recovery exhibits  two distinct characteristics depending on the number of the measurement vectors. First, if  the number of multiple measurement vectors $r$ is sufficiently small, then the number of samples for S-OMP is reciprocally proportional to the number of multiple measurement vectors. On the other hand, we have sufficiently large number of snapshots such that $\lim_{n\rightarrow\infty}(\log{n})/r$ is close to 0, then the number of measurements for S-OMP varies from $4k$ to $k$ according to the ratio of $r$ and $k$ so that the $\log{n}$ is not necessary. In particular,  if the number of snapshots approaches the sparsity $k$, then we can identify the indices of ${\rm supp}X$ with only $k$ measurements, which is equivalent to the required number of multiple measurement vectors for the success of conventional MUSIC. Furthermore, we demonstrated that the required SNR for the success of support recovery can be reduced and  when the asymptotic ratio of the number of snapshots and the sparsity level (that is, $\lim_{n\rightarrow\infty}r/k$) is nonzero in the large system limit, only finite SNR is required, which is significant improvement over SMV-CS.

In the original form of CS-MUSIC, the performance is, however,  very dependent on the selection of $k-r$ correct indices of the support of $X$. In practice, even though the consecutive $k-r$ steps of S-OMP may not be correct, there are chances that among the estimates of $k$-sparse solution, part of the supports could be correct. Hence, if we have a mean to identify $k-r$ correct support in any order out of any $k$-sparse, then we can expect that the performance of the compressive MUSIC will be improved. Of course, when ${k \choose k-r}$ is small, we may apply the exhaustive search, but if both $k-r$ and $r$ are not small, then the exhaustive search is hard to apply so that we have to find some alternative method to identify the correct indices from the estimate of ${\rm supp}X$.
Indeed, the following support selection criterion can address the problem \cite{Kim2011CSMUSIC}.
\begin{theorem}\cite{Kim2011CSMUSIC}\label{thm-sfcriterion}
Assume that we have a canonical MMV model $AX=B$ where $A\in\mathbb{R}^{m\times n}$, $X\in\mathbb{R}^{n\times r}$, $\|X\|_0=k$ and $r<k<m<n$. If there is an index set $I_k\subset \{1,\cdots,n\}$ such that $|I_k|=\min\{k, {\rm spark}(A)-r\}$ and $|I_k\cap {\rm supp}X|\geq k-r+1$, then for any $j\in I_k$,
$j\in {\rm supp}X$ if and only if
\begin{equation}\label{sf-criterion}
P_{Q_{k,j}}\ab_j={\bf 0},
\end{equation}
where $Q_{k,j}$ is the orthogonal complement for $R([B~A_{I_k\setminus \{j\}}])$, $A_{I_k\setminus \{j\}}$ consists of columns of $A$ whose index belongs to $I_k\setminus \{j\}$ and $P_{R([B~~A_{I_k\setminus \{j\}}])}^{\perp}$ is the orthogonal projection on $R([B~~A_{I_k\setminus \{j\}}])^{\perp}$.
In particular, if the columns of $A$ are in general position, then we can take index set $I_k$ with $|I_k|=\min\{k,m-r+1\}$. Also, if $A$ has an RIP condition with $0<\delta_{2k}<1$, then we can take $|I_k|=k$ since $r\leq k$. 
\end{theorem}

Theorem~\ref{thm-sfcriterion} informs us that we only require the success of partial support recover out of $k$-sparse estimate, rather than $k-r$ consecutive correct CS step \cite{Kim2010CMUSIC}.
Accordingly, the compressive MUSIC with optimized partial  support is then performed by following procedure.

\begin{itemize}
\item  {\bf [Step 1: compressed sensing]} Estimate $k$ indices of ${\rm supp}X$ by any MMV compressive sensing algorithm.
Let $I_{k}$ be the set of indices which are taken in step 1.
\item  {\bf [Step 2: support deletion]}  For $j\in I_{k}$, calculate the quantities $\zeta(j)=
\|P_{Q_{k,j}}\ab_j\|^2.$
Make an ascending ordering of $\zeta(j)$, $j\in I_k$ and choose indices that corresponds the first $k-r$ elements and put these indices into $S$ and remove the remaining ones.
\item  {\bf [Step 3: support addition]} For $j\in \{1,\cdots,n\}\setminus S$, calculate the quantities
$$\eta(j)=\ab_j^{*}[P_{R(Q)}-P_{R(P_{R(Q)}A_{I_{k-r}})}]\ab_j.$$
 Make an asending ordering of $\eta(j)$, $j\notin S$ and choose indices that correspond to the first $r$ elements and put these indices into $S$.
\end{itemize}
The step $1$ in the above algorithm need not to be  greedy so that we can also apply the convex optimization algorithm such as $l_{2,1}$ minimization \cite{malioutov2005ssr} or belief propagation \cite{KCJBY2011}.

\section{Deterministic Support Tracking using Compressive MUSIC}
\label{sec:proposedMethod}

\subsection{Noiseless Cases}

In this section, we will show how the compressive MUSIC with optimized partial support can be used for dynamic support tracking,  whose joint support ${\rm supp}X(t)$ changes slowly along time as illustrated in Fig.\ref{fig:timevaryingMMV}.
First, we define a canonical form of dynamic MMV problem.

\begin{definition}\label{def:can}
A canonical form of noiseless dynamic MMV problem is given by set of MMV problem with time varying $k$-sparse vectors $X(t)\in \mathbb{R}^{n \times r}$ that satisfies $Y(t)=AX(t)$ as described in following formulation:
\begin{eqnarray}\label{eq:timevaryingMMV}
\min\limits_{X(t)}~\|X(t)\|_0,~~{\rm subject~to}~~B(t)=AX(t),~~t=0,1,\cdots,
\end{eqnarray}
where ${\rm supp}X(t)=\{1\leq i\leq n : \mathbf{x}(t)^i\neq 0\}$ and   $|{\rm supp}X(t)|=k(t)$, the measurement matrix $B(t)$ is full rank,  i.e. ${\rm rank}(B(t))\leq k(t)$. Here we assume that ${\rm rank}(B(t))$ is constant so that we let $r:={\rm rank}(B(t))$. 
\end{definition}

Note that the canonical form MMV has the additional
constraints that the measurement matrix is full rank and  ${\rm rank}(B(t))=r\leq k(t)$. This is not
problematic  since every dynamic MMV problem  can be converted into
a canonical form using the following dimension reduction similar to  \cite{Kim2010CMUSIC} .
\begin{itemize}
   \item Suppose we are given the following linear sensor observations: $B(t)=AX(t)$ where  $A\in\mathbb{R}^{m\times n}$ and
   $X\in\mathbb{R}^{n\times l}$ satisfies $\|X(t)\|_0=k(t)$.
   \item Compute the SVD as $B(t)=UD_rV^{*}$, where $D_r$ is an $r\times r$ diagonal matrix, $V\in\mathbb{C}^{l\times r}$ consists of right singular vectors,
    and $r={\rm rank}(B)$, respectively.
   \item Reduce the dimension as $B_{SV}(t)=B(t)V$ and $X_{SV}(t)=X(t)V$.
   \item The resulting canonical form MMV becomes $B_{SV}(t)=AX_{SV}(t)$.
\end{itemize}
We can easily show that  ${\rm rank}(B_{SV})=r\leq k(t)$  and full rank and
 the sparsity $k(t):=\|X(t)\|_0 = \|X_{SV}(t)\|_0$ with probability 1.
Therefore, without loss of generality,  the canonical form of dynamic MMV in Definition~\ref{def:can}
is assumed throughout the paper.

\begin{figure}
   \begin{center}
   \begin{tabular}{c}
   \includegraphics[width=11cm, height=3.5cm]{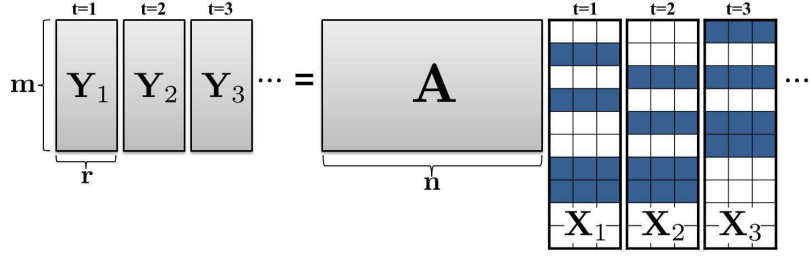}
   \end{tabular}
   \end{center}
   \caption[example]
   { \label{fig:timevaryingMMV}
MMV problem for slowly time varying sparsity pattern.}
\end{figure}

For such dynamic support tracking,  we can apply our CS-MUSIC algorithm. However,
 if the number of snapshots is not sufficient, the amount of support estimation that need to be done by CS step is significantly larger than those recovered by the deterministic generalized MUSIC step. Since CS step allows the support recovery in a probabilistic sense,  it is more prone to error;  so  we are interested in finding a deterministic algorithm that significantly outperform the existing one.  The following Theorem \ref{thm:sparsity-varying} shows that if we have a  correct estimation for the initial support $I(0)$ of $X(0)$ and the support changes are sufficiently small and the sparsity $k(t)$ is fixed for all time point, then we can recursively identify the support of time-varying input signals in a {\em deterministic} manner.


\begin{theorem}\label{thm:sparsity-varying}
Suppose a noiseless canonical form of dynamic MMV problem  satisfies
\begin{equation}\label{slowly-time-varying}
|{\rm supp}X(t)\setminus {\rm supp}X(t-1)|\leq r-1 ,
\end{equation}
for all $t=1,2, \cdots$. Furthermore we assume that $r\leq k(t)\leq k_{\max}$ for a positive integer $k_{\max}$ and $0\leq \delta_{2k_{\max}}(A)<1$.
Then, if we have a correct initial support estimation for $X(0)$, then we can identify the correct support for all $t>0$ by applying the following procedure recursively: 
\begin{itemize}
\item {\bf [Initial support estimation]} Let $I(t-1)$ be the support estimation of $X(t-1)$;
\item {\bf [Support deletion]} Find an index set $I(t)^{a}\subset I(t-1)$ such that 
$I(t)^{a}:=\{j\in I(t-1):\ab_j^{*}P_{Q(t)_{k,j}}\ab_j=0\},$
where $Q(t)_{k.j}$ is the orthogonal complement for \newline  $R[B(t),A_{I(t-1)\setminus \{j\}}]$;
\item {\bf [Support addition]} Find an index set $I(t)$ such that \\
$I(t)=\{j:\ab_j^*[P_{R(Q(t))}-P_{R(P_{R(Q(t))A_{I(t)^a}})}]\ab_j=0\},$
where $Q(t)\in\mathbb{R}^{m\times (m-r)}$ consists of orthonormal columns such that $Q(t)^{*}B(t)=0$;
\item Set $\hat{k}(t):=|I(t)|$ be the sparsity estimate for $X(t)$ and $I(t)$ be the support estimate for $X(t)$.
\end{itemize}
\end{theorem}
\begin{proof}
See Appendix~A.
\end{proof}
In Theorem \ref{thm:sparsity-varying}, we assume the RIP condition $0\leq \delta_{2k_{\max}}^L(A)<1$, instead of $0\leq \delta_{2k_{\max}-r+1}^L(A)<1$. If we assuming the RIP condition $0\leq \delta_{2k_{\max}-r+1}^L(A)<1$, when $r>1+k_{\max}-k(t)/2$, we may have $|I_k|<k(t)$. However, we can modify the support deletion procedure in Theorem \ref{thm:sparsity-varying} as the following, under the condition $|{\rm supp}X(t)\setminus {\rm supp}X(t-1)|\leq k_{\max}/2$.
\begin{itemize}
   \item {\bf [Support deletion]} Find an index set $I(t)^{a}\subset I(t-1)$ such that 
$I(t)^{a}:=\{j\in I(t-1):\ab_j^{*}P_{Q(t)_{k,j}}\ab_j=0\},$
where $Q(t)_{k.j}$ is the orthogonal complement for $R[\tilde{B}(t),A_{I(t-1)\setminus \{j\}}]$  and $\tilde{B}(t)$ consists of $1+[\frac{k_{\max}}{2}]$ columns of $B(t)$.
\end{itemize}

\subsection{Noisy Cases}

 In practice, the measurements are noisy, so the theory we derived for noiseless measurement should be modified. 
 In the noisy case, when the sparsity are known {\it a priori} and does not change along time, we can apply the following procedure.
\begin{itemize}
    \item Let $t=0$ and let $I(0)$ be the support estimation of $X(0)$.
    \item  For all $t=1,2,\cdots$, do
        \begin{itemize}
            \item Let $I(t)=\emptyset.$
            \item For all $j\in I(t-1)$, calculate the quantities $\zeta(j)=\|P_{Q(t)_{j,k}}\ab_j\|^2$.
            \item Make an ascending ordering of $\zeta(j)$ and choose indices that correspond to the first $k-r$ elements and put these indices into $I(t)$.
            \item For $j\in \{1,\cdots,n\}\setminus I(t)$, calculate the quantities $\eta(j)=\ab_j^*\left[P_{R(Q(t))}-P_{R(P_{R(Q(t))A_{I(t)}})}\right]\ab_j.$
            \item Make an ascending ordering of $\eta(j)$, $j\notin I(t)$ and choose indices that correspond to the first $r$ indices and add these indices to $I(t)$.
            \item $I(t)$ is the estimation of ${\rm supp}X(t)$ and let $t=t+1$.
        \end{itemize}
\end{itemize}

However,  if the sparsity changes along time, in the noisy cases, some of the steps in Theorem \ref{thm:sparsity-varying} should be modified as follows:
\begin{itemize}
\item    {\bf [Support deletion]}  Set $\epsilon_1>0$ and find an index set $I(t)^a$ such that 
$I(t)^a=\{j\in I(t-1):\ab_j^{*}P_{Q(t)_{k,j}}\ab_j<\epsilon_1\}$
where $Q(t)_{k,j}$ is the orthogonal complement for $R[Y(t)~A_{I_1(t)\setminus\{j\}}]$, where $I_1(t) \subset I(t-1)$ such that $${\rm nrank}[Y(t)~A_{I_1(t)}]={\rm nrank}[Y(t)~A_{I(t-1)}]=r+|I_1(t)|,$$ where ${\rm nrank}(A)$ denotes the numerical rank of $A$.
\item  {\bf [Support addition]}   Set $\epsilon_2>0$ and find an index set $I(t)^b$ such that 
$$I(t)^b=\{j\notin I(t)^a:\ab_j^{*}P_{R([Y(t)~A_{I_2(t)}])^{\perp}}\ab_j<\epsilon_2\},$$ where an index set $I_2(t)\subset I(t)^a$ such that $${\rm nrank}[Y(t)~A_{I_2(t)}]={\rm nrank}[Y(t)~A_{I(t)^a}]=r+|I_2(t)|.$$
\end{itemize}
 In this section, we
derive sufficient conditions for  the threshold values and signal to noise ratio
 that
guarantee the correct identification of time varying support.  
For CS-MUSIC  \cite{Kim2010CMUSIC}, we derived an expression of SNR and the 
minimum number of sensor elements. 
Even though these derivation is based on a large system model with
a Gaussian sensing matrix, it has provided very useful insight. Therefore, we employed a
large system model to derive a sufficient condition for  the success of proposed algorithm.

\begin{definition}\label{def:d-mmv}
A large system noisy canonical form of dynamic MMV  is defined as an estimation
problem of $k(t)$-sparse vectors $X(t)\in\mathbb{R}^{n\times r}$ that
shares a common sparsity pattern through multiple noisy snapshots
$Y(t)=AX(t)+N(t)$ using the following formulation:
\begin{eqnarray}\label{mmv-thres}
{\rm minimize}~~~\|X(t)\|_0\\
{\rm subject~to}~~~Y(t)=AX(t)+N(t), \notag
\end{eqnarray}
where $A\in\mathbb{R}^{m\times n}$ is a random matrix with i.i.d.
$\mathcal{N}(0,1/m)$ entries, $N
=[\nb_1,\cdots,\nb_r]\in\mathbb{R}^{m\times r}$ is  an additive
noise matrix, $m\rightarrow\infty$, $k\rightarrow\infty$ as $n\rightarrow\infty$ and ${\rm rank}(AX(t))=r(t)\leq k(t)=\|X(t)\|_0$. \ Here, we
assume that $\rho:=\lim_{n\rightarrow\infty}m/n>0$ and
$\gamma=\lim_{n\rightarrow\infty}k_{\max}/m> 0$, $\alpha:=\lim_{n\rightarrow\infty}r/k_{\max}\geq 0$ exist and $\alpha\leq 1-\epsilon$ for some $0<\epsilon<1$.
\end{definition}

Under the large system model,  we have the following theorem. 

\begin{theorem}
Consider the large system model dynamic MMV in Definition~\ref{def:d-mmv}. Suppose a minimum {\sf SNR} satisfies 
\begin{equation}\label{snr-supp-selection}
{\sf SNR}_{\min}(Y(t)):=\frac{\sigma_{\min}(B(t))}{\|N\|}> 1+\frac{4(\kappa(B(t))+1)}{1-\gamma(1+\alpha)},
\end{equation}
where $\sigma_{\min}(B(t))$ is the minimum singular value for $B(t)$, $\|N\|$ is the spectral norm of $N\in\mathbb{R}^{m\times r}$ and $B(t)$ is the noiseless measurements, and $\alpha = \lim_{n\rightarrow\infty}r/k_{\max}$, $\gamma=\lim_{n\rightarrow\infty}k_{\max}/m$. Then,
 for the noisy canonical form dynamic MMV problem for slowly time varying pattern that satisfies Eq.~\eqref{slowly-time-varying}, the threshold values for support deletion and addition  criterion to the correct partial support for $X(t)$ are given by 
 \begin{eqnarray}
\epsilon_1:=(1-\gamma(1+\alpha))/2, & \quad \epsilon_2:=(1-\gamma)/2.
\end{eqnarray}
\end{theorem}
\begin{proof}
See Appendix~B.
\end{proof}

%
%
%


\section{NUMERICAL RESULTS}
\label{sec:results}

The first simulation is to demonstrate the performance of the proposed method to solve the time varying MMV problem in Eq.~\eqref{eq:timevaryingMMV} for different number of changes in supports at each time. We declared the algorithm as a success if the estimated support is the same as the true ${\rm supp}X$, and the success rates were averaged for $5000$ experiments. The simulation parameters were as follows: $m=40$, $n=100$, $r=9$, and $k\in \{ 1,2,\cdots,30 \}$, respectively. Elements of sensing matrix $A$ were generated by i.i.d. Gaussian random variable $\frac{1}{\sqrt{m}}\mathcal{N}(0,1)$, and Gaussian noise of ${\rm SNR}=40dB$ was added to each measurement vectors. At each time point, $X(t)^{{\rm supp}X(t)}$ is generated by $\mathcal{N}(0,1)$. Fig.\ref{fig:tracking_graph} shows the recovery rates of time varying MMV problem using support tracking method for $t=1,2,\cdots,5$ when the number of changed supports are $4,6,7$, and $8$ at each time point for Fig.\ref{fig:tracking_graph}(a)$\sim$(d), respectively. We used CS-MUSIC algorithm with S-OMP and then applied optimized partial support selection at $t=1$, and time varying supports are estimated by support tracking method recursively from $t=2$ to $t=5$. In Fig.\ref{fig:tracking_graph}, we can observe that the performance gracefully decreases as the number of changes in supports increases. An interesting observation is that the performance of the proposed method rather improves over time in Fig.\ref{fig:tracking_graph}(a) and (b). However, the recovery ratio is getting lower but converges over time when the number of changes in supports is close to the upper bound $r-1$ for perfect recovery in noiseless case.

\begin{figure}
   \begin{center}
   \begin{tabular}{c}
   \includegraphics[width=10cm, height=9cm]{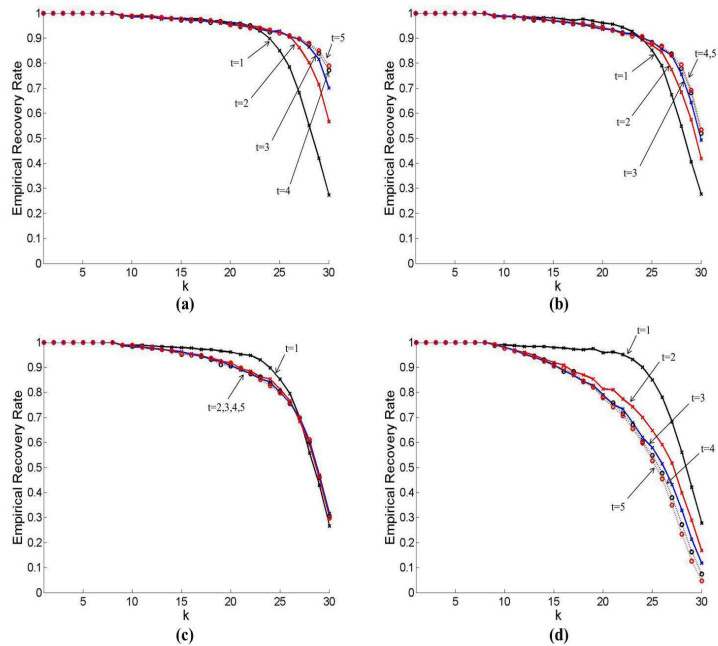}
   \end{tabular}
   \end{center}
   \caption[example]
   { \label{fig:tracking_graph}
Recovery rates of time varying MMV problem using support tracking method when $m=40$, $n=100$, $r=9$, SNR$=40$dB, and $t=1,2,\cdots,5$. The number of changes in supports at each time point is (a) $4$, (b) $6$, (c) $7$, and (d) $8$.}
\end{figure}

Next, we applied the proposed algorithm to target tracking problem in $2D$ image and compared it to MUSIC algorithm. The first row of Fig.\ref{fig:tracking_2D} indicates the original targets moving toward the direction of red arrows over time. Each column (from left to right) indicates the sampled image at $t=1,13,27$, and $t=41$, respectively. The simulation setting is the same with the previous one except $m=50$, $n=900$, $t=1,2,\cdots,45$, and each target have a chance to move with probability of $\frac{1}{2}\frac{r-1}{k}$ at each time point. The number of target $k$ is $24$. Here, we considered the number of measurement vectors is $50$ in the resting state, and used MUSIC algorithm to find supports at $t=0$. The second and third row of Fig.\ref{fig:tracking_2D} indicate the results of support tracking method and MUSIC algorithm, respectively. Note that the proposed method successfully follows the movement of original targets.

\begin{figure}
   \begin{center}
   \begin{tabular}{c}
   \includegraphics[width=13cm, height=9cm]{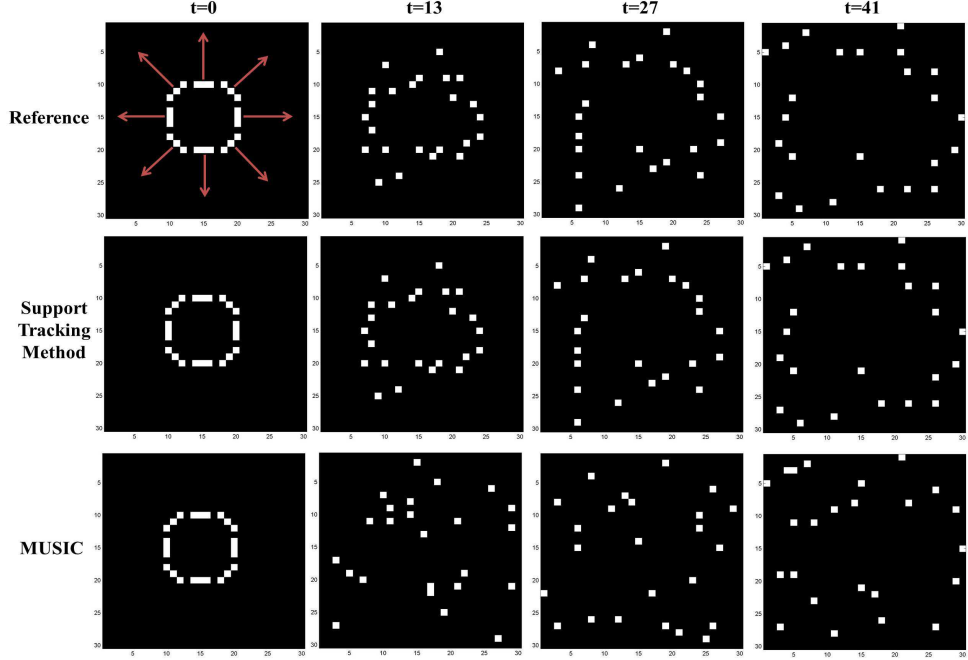}
   \end{tabular}
   \end{center}
   \caption[example]
   { \label{fig:tracking_2D}
The results of the target tracking problem in $2D$ image when $m=50$, $n=900$, $k=24$, and SNR$=40$dB. We set $r=50$ when $t=0$, and $r=9$ for $t>0$. The first row indicates the original targets moving toward the direction of red arrows over time. The second and third row indicate the results of support tracking method and MUSIC algorithm, respectively. Each column (from left to right) indicates the sampled image at $t=1,13,27$, and $t=41$, respectively.}
\end{figure}

\section{CONCLUSION}
\label{sec:conclusion}

This paper expanded the sparse recovery with partially known supports in single measurement vector problem to multiple measurement vector problem with joint sparsity and proposed the support tracking algorithm to recover the slowly time varying supports. It is based on the recently developed compressive MUSIC algorithm with optimized partial support selection. The estimated supports at previous time can be used in optimized partial support selection to recover partial supports at current time and it can be used in generalized MUSIC criterion to find remaining supports. We also provided the maximum allowable number of changes in supports with support tracking algorithm for exact reconstruction in noiseless case. Numerical results demonstrated that the proposed algorithm reliably reconstructs the time varying supports for various level of changes and successfully solves the target tracking problem in $2D$ image.

\section*{Appendix~A}
\label{ap:A}

\begin{proof}
We only need to show that if we have a correct support for $X(t-1)$, then we can also obtain a correct support estimation for $X(t)$ by the support selection criterion and the generalized MUSIC criterion. 
By the assumption, we have $m\geq 2k_{\max}\geq k(t-1)+r$ so that if we have 
$|{\rm supp}X(t)\cap I(t-1)|\geq k(t)-r+1$, then by Theorem \ref{thm-sfcriterion} we have for any $j\in I(t-1)$, 
$$j\in {\rm supp}X(t)~{\rm if~and~only~if}~\ab_j^{*}P_{Q(t)_{k,j}}\ab_j=0$$
where $Q(t)_{k,j}$ is the orthogonal complement of $R([Y(t)~A_{I(t-1)\setminus\{j\}})$. Since we have a noiseless MMV problem with slowly time varying pattern, we have $|{\rm supp}X(t)\setminus {\rm supp}X(t-1)|\leq r-1$ so that we have 
$|{\rm supp}X(t)\cap I(t-1)|\geq k(t)-r+1$ and we can identify the correct partial support of $X(t)$ which has at least $k(t)-r+1$ elements. Then, if we let $I(t)^a$ be the set of indices such that 
$$I(t)^a=\{j\in I(t-1):\ab_j^{*}P_{Q(t)_{k,j}}\ab_j=0\},$$
we have $I(t)^a\subset {\rm supp}X$ and $R([Y(t)~A_{I(t)^a}])\subset R(A_{{\rm supp}X(t)})$. On the other hand, if we take a set $I(t,r)\subset I(t)^a\subset {\rm supp}X$ such that $|I(t,r)|=k(t)-r$, we have 
$$R([Y(t)~A_{I(t)^a}])\supset R([Y(t)~A_{I(t,r)}])=R(A_{{\rm supp}X(t)})$$
which implies $R([Y(t)~A_{I(t)^a})=R(Y(t)~A_{I(t,r)})$. 
Since $0\leq\delta_{2k(t)-r+1}(A)\leq \delta_{2k_{\max}-r+1}(A)<1$, we can apply the generalized MUSIC criterion with $I(t,r)\subset {\rm supp}X$ where $|I(t,r)|=k(t)-r$. 
For $j\in I(t,r)$, we can easily see that 
$$\ab_j^*\left[P_{R(Q(t))}-P_{R(P_{R(Q(t))A_{I(t)^a}})}\right]\ab_j
=\ab_j^{*}P_{R([Y~A_{I(t)^a}])^{\perp}}\ab_j=\ab_j^{*}P_{R([Y~A_{I(t,r)}])^{\perp}}\ab_j=0.$$ On the other hand, for $j\notin I(t,r)$, by the generalized MUSIC criterion, we have $j\in {\rm supp}X(t)$ if and only if 
$$\ab_j^*\left[P_{R(Q(t))}-P_{R(P_{R(Q(t))A_{I(t)^a}})}\right]\ab_j=\ab_j^*\left[P_{R(Q(t))}-P_{R(P_{R(Q(t))A_{I(t,r)}})}\right]\ab_j=0.$$
Since $I(t,r)\subset {\rm supp}X$, we have $j\in {\rm supp}X$ if and only if 
$$\ab_j^*\left[P_{R(Q(t))}-P_{R(P_{R(Q(t))A_{I(t)^a}})}\right]\ab_j=0.$$
Hence, $|I(t)|=k(t)$ and $I(t)={\rm supp}X(t)$.
\end{proof}

\section*{Appendix~B}
\begin{proof}
Here, we let $B(t)=AX(t)$, $\sigma_{\min}(B(t))$(or $\sigma_{\min}(B(t))$) be the minimum (or the maximum) nonzero singular value of $B(t)$. Then $Y(t)=B(t)+N(t)$ is also of full column rank if $\|N(t)\|<\sigma_{\min}(B(t))$. By \cite{Kim2010CMUSIC}, for such an $N(t)$, we have 
\begin{equation}\label{proj-diff}
\|P_{R(Y(t))}-P_{R(B(t))}\|\leq \frac{2[\sigma_{\max}(B(t))+\sigma_{\min}(B(t))]\|N(t)\|}{\sigma_{\min}(B(t))(\sigma_{\min}(B(t))-\|N(t)\|)}.
\end{equation} 
By the projection update rule,  we have 
\begin{equation}\label{proj-update1}
P_{R([B(t)~A_{I_1(t)\setminus\{j\}}])}=P_{R(A_{I_1(t)\setminus\{j\}})}+
P_{R(P_{R(A_{I_1(t)\setminus\{j\}})}^{\perp}B(t))}
\end{equation}
and
\begin{equation}\label{proj-update2}
P_{R([Y(t)~A_{I_1(t)\setminus\{j\}}])}=P_{R(A_{I_1(t)\setminus\{j\}})}+
P_{R(P_{R(A_{I_1(t)\setminus\{j\}})}^{\perp}Y(t))}.
\end{equation}
Since $[B(t)~A_{I_1(t)\setminus\{j\}}]$ and $[Y(t)~A_{I_1(t)\setminus\{j\}}]$ are of full column rank,  by applying \eqref{proj-update1} and \eqref{proj-update2} as done in \cite{Lee2010SAMUSIC}, we have 
\begin{eqnarray}\label{perturb-diff}
\|P_{R([B(t)~A_{I_1(t)\setminus\{j\}}])}^{\perp}-P_{R([Y(t)~A_{I_1(t)\setminus\{j\}}])}^{\perp}\|
&=&\|P_{R(P_{R(A_{I_1(t)\setminus\{j\}})}^{\perp}B(t))}-P_{R(P_{R(A_{I_1(t)\setminus\{j\}})}^{\perp}Y(t))}\|\notag\\
&\leq&\|P_{R(B(t))}-P_{R(Y(t))}\|. 
\end{eqnarray}
Then for any $j\in I_1(t)\setminus{\rm supp}X$, we have 
\begin{eqnarray}\label{j-not-supp}
\ab_j^{*}P_{R([Y(t)~A_{I_1(t)\setminus\{j\}}])}^{\perp}\ab_j&=&\ab_j^{*}P_{R([B(t)~A_{I_1(t)\setminus\{j\}}])}^{\perp}\ab_j\notag\\
&&+~\ab_j^{*}\left[
P_{R([Y(t)~A_{I_1(t)\setminus\{j\}}])}^{\perp}-P_{R([B(t)~A_{I_1(t)\setminus\{j\}}])}^{\perp}\right]\ab_j\\
&\geq&\min_{j\notin{\rm supp}X}\ab_j^{*}P_{R([B(t)~A_{I_1(t)\setminus\{j\}}])}^{\perp}\ab_j-\max\limits_{1\leq j\leq n}\|\ab_j\|^2\|P_{R(Y(t))}-P_{R(B(t))}\|.\notag
\end{eqnarray}
Here, for each $1\leq j\leq n$, $m\|\ab_j\|^2$ is a chi-square random variable with degree of freedom $m$ so that we have by Lemma 3 in \cite{Fletcher2009nscond}, $\lim_{n\rightarrow\infty}\max_{1\leq j\leq n}\|\ab_j\|^2=1$ since $\lim_n (\log{n})/m=0$. Furthermore, for any $j\notin {\rm supp}X$, $\ab_j$ is independent of $P_{R([Y(t)~A_{I(t-1)\setminus\{j\}}])}^{\perp}$, so that $m\ab_j^{*}P_{R([Y(t)~A_{I_1(t)\setminus\{j\}}])}^{\perp}\ab_j$ is a chi-squared random variable whose degree of freedom is at least $m-k(t)-r+1$ since $P_{R([Y(t)~A_{I_1(t)\setminus\{j\}}])}^{\perp}$ is a projection operator onto the orthogonal couplement of $R([Y(t)~A_{I_1(t)\setminus\{j\}}])$. Since $\lim_{n\rightarrow\infty}
(\log{(n-k(t))})/(m-k(t)-r+1)=0$, again by Lemma 3 in \cite{Fletcher2009nscond}, we have 
$$\lim\limits_{n\rightarrow\infty}\frac{\min\limits_{j\notin{\rm supp}X}m\ab_j^{*}
P_{R([Y(t)~A_{I_1(t)\setminus\{j\}}])}^{\perp}\ab_j}{m-k(t)-r+1}\geq1$$ 
so that 
\begin{equation}\label{j-not-supp-1}
\lim\limits_{n\rightarrow\infty}\min\limits_{j\notin{\rm supp}X}\ab_j^{*}
P_{R([Y(t)~A_{I_1(t)\setminus\{j\}}])}^{\perp}\ab_j\geq 1-\gamma(1+\alpha)
\end{equation}
since $k(t)\leq k_{\max}$ for all $t=0,1,\cdots$.
On the other hand, if we use the definition of ${\sf SNR}_{\min}(Y(t))$ and the definition of the condition number of $B(t)$ on \eqref{proj-diff}, i.e. $\kappa(B(t))=(\sigma_{\max}(B(t)))/(\sigma_{\min}(B(t)))$, we have 
\begin{equation}\label{j-not-supp-2}
\|P_{R(Y(t))}-P_{R(B(t))}\|\leq \frac{2(\kappa(B(t))+1)}{{\sf SNR}_{\min}(Y(t))-1}<\frac{1-\gamma(1+\alpha)}{2},
\end{equation}
by the condition \eqref{snr-supp-selection}. Combining \eqref{j-not-supp}, \eqref{j-not-supp-1} and \eqref{j-not-supp-2}, we have for any $j\in I(t-1)\setminus {\rm supp}X$, we have 
\begin{equation*}
\ab_j^{*}P_{R([Y(t)~A_{I_1(t)\setminus\{j\}}])}^{\perp}\ab_j>\frac{1-\gamma(1+\alpha)}{2}.
\end{equation*}
On the other hand, for $j\in I(t-1)\cap {\rm supp}X(t)$, we have $\ab_j^{*}P_{R([B(t)~A_{I_1(t)\setminus\{j\}}])}^{\perp}\ab_j=0$ by the support selection criterion. Then, by the similar reasoning as above, we have for any $j\in I_1(t)\cap {\rm supp}X(t)$, we have 
\begin{equation*}
\ab_j^{*}P_{R([Y(t)~A_{I_1(t)\setminus\{j\}}])}^{\perp}\ab_j<\frac{1-\gamma(1+\alpha)}{2}.
\end{equation*}
This completes the proof for the threshold values for support selection criterion. The proof for the threshold values for generalized MUSIC are the same except that $m\ab_j^{*}P_{R([Y(t)~A_{I_2(t)}])}^{\perp}\ab_j$ is a chi-squared random variable whose degree of freedom is $m-k(t)$ for $j\notin {\rm supp}X(t)$. 
\end{proof}

\section*{Acknowledgment}

This work was supported by the Korea Science and Engineering Foundation (KOSEF) grant funded by the Korea government (MEST) (No.2010-0000855).

\ifCLASSOPTIONcaptionsoff
  \newpage
\fi

\bibliographystyle{IEEEbib}
\bibliography{totalbiblio_bispl}

\begin{thebibliography}{10}

\bibitem{Zhang2011TSBL}
Z.~Zhang and B.D. Rao,
\newblock ``Sparse signal recovery with temporally correlated source vectors
  using joint sparse bayesian learning,''
\newblock {\em IEEE J. of Selected Topics in Signal Processing}, vol. 5, no. 5,
  pp. 912--926, 2011.

\bibitem{Do06}
D.~L. Donoho,
\newblock ``Compressed sensing,''
\newblock {\em IEEE Trans. on Information Theory}, vol. 52, no. 4, pp.
  1289--1306, April 2006.

\bibitem{CaRoTa06}
E.~Candes, J.~Romberg, and T.~Tao,
\newblock ``Robust uncertainty principles: Exact signal reconstruction from
  highly incomplete frequency information,''
\newblock {\em IEEE Trans. on Information Theory}, vol. 52, no. 2, pp.
  489--509, Feb. 2006.

\bibitem{Vaswani2010MCS}
N.~Vaswani and W.~Lu,
\newblock ``Modified-{CS}: modifying compressive sensing for problems with
  partially known support,''
\newblock {\em IEEE Trans. Signal Process.}, vol. 58, no. 9, pp. 4595--4607,
  2010.

\bibitem{CaTa05}
E.~Candes and T.~Tao,
\newblock ``Decoding by linear programming,''
\newblock {\em IEEE Trans. on Information Theory}, vol. 51, no. 12, pp.
  4203--4215, Dec. 2005.

\bibitem{jin2011support}
Y.~Jin and B.D. Rao,
\newblock ``Support recovery of sparse signals in the presence of multiple
  measurement vectors,''
\newblock {\em arXiv preprint}, 2011,
\newblock http://arxiv.org/pdf/1109.1895.

\bibitem{Kim2010CMUSIC}
J.~M. Kim, O.~K. Lee, and J.~C. Ye,
\newblock ``Compressive {MUSIC}: a missing link between compressive sensing and
  array signal processing,''
\newblock {\em to appear in IEEE Trans. Inf. Theory}, 2011.

\bibitem{Kim2011CSMUSIC}
J.~M. Kim, O.~K. Lee, and J.~C. Ye,
\newblock ``Compressive {MUSIC} with optimized partial support for joint sparse
  recovery,''
\newblock {\em in Proc. IEEE Int. Symp. Inf. Theory (ISIT)}, 2011.

\bibitem{Donoho20204npss}
D.~L. Donoho,
\newblock ``{Neighborly polytopes and sparse solution of underdetermined linear
  equations},''
\newblock {\em Tech. report, Department of Statistics, Stanford University},
  2005.

\bibitem{DoEl03}
D.~L. Donoho and M.~Elad,
\newblock ``Optimally sparse representation in general (non-orthogonal)
  dictionaries via $l_1$ minimization,''
\newblock {\em Proceedings of the National Academy of Sciences of the United
  States of America}, vol. 100, no. 5, pp. 2197--2202, 2003.

\bibitem{chen2006trs}
J.~Chen and X.~Huo,
\newblock ``{Theoretical results on sparse representations of multiple
  measurement vectors},''
\newblock {\em IEEE Trans. on Signal Processing}, vol. 54, no. 12, pp.
  4634--4643, 2006.

\bibitem{Feng97}
P.~Feng,
\newblock {\em Universal minimum-rate sampling and spectrum-blind
  reconstruction for multiband signals},
\newblock Ph.{D}. dissertation, University of Illinois, Urbana-Champaign, 1997.

\bibitem{Davies2010Rank}
M.~E. Davies and Y.~C. Eldar,
\newblock ``Rank awareness for joint sparse recovery,''
\newblock {\em preprint}, 2010,
\newblock http://arxiv.org/PS\_cache/arxiv/pdf/1004/1004.4529v1.pdf.

\bibitem{schmidt1986multiple}
R.~Schmidt,
\newblock ``{Multiple emitter location and signal parameter estimation},''
\newblock {\em IEEE Trans. on Antennas and Propagation}, vol. 34, no. 3, pp.
  276--280, 1986.

\bibitem{gribonval2008aac}
R.~Gribonval, H.~Rauhut, K.~Schnass, and P.~Vandergheynst,
\newblock ``{Atoms of all channels, unite! Average case analysis of
  multi-channel sparse recovery using greedy algorithms},''
\newblock {\em Journal of Fourier Analysis and Applications}, vol. 14, no. 5,
  pp. 655--687, 2008.

\bibitem{Eldar2010aca}
Y.C. Eldar and H.~Rauhut,
\newblock ``{Average case anlysis of multichannel sparse recovery using convex
  relaxation},''
\newblock {\em IEEE Trans. on Information Theory}, vol. 56, pp. 505--519, 2010.

\bibitem{Lee2010SAMUSIC}
K.~Lee and Y.~Bresler,
\newblock ``Subspace-augmented music for joint sparse recovery,''
\newblock {\em preprint}, 2010,
\newblock http://arxiv.org/PS\_cache/arxiv/pdf/1004/1004.3071.pdf.

\bibitem{malioutov2005ssr}
D.~Malioutov, M.~Cetin, and AS~Willsky,
\newblock ``{A sparse signal reconstruction perspective for source localization
  with sensor arrays},''
\newblock {\em IEEE Trans. on Signal Processing}, vol. 53, no. 8, pp.
  3010--3022, 2005.

\bibitem{KCJBY2011}
J.M. Kim, W.H. Chang, B.C. Jung, D.~Baron, and J.C. Ye,
\newblock ``Belief propagation for joint sparse recovery,''
\newblock {\em arXiv preprint}, 2011,
\newblock http://arxiv.org/pdf/1102.3289.

\bibitem{Fletcher2009nscond}
S.~Rangan A.K.~Fletcher and V.K. Goyal,
\newblock ``Necessary and sufficient conditions for sparsity pattern
  recovery,''
\newblock {\em IEEE Trans. on Inform. Theory}, vol. 55, no. 12, pp. 5758--5772,
  December 2009.

\end{thebibliography}

\end{document}